\documentclass[11pt,a4paper,reqno]{amsart}
\usepackage{amsfonts}
\usepackage{amssymb}
\usepackage[dvips]{graphics}
\usepackage{epsfig}
\usepackage{euscript}
\usepackage{color}
\usepackage{amsmath}
\usepackage[T1]{fontenc}
\usepackage[utf8x]{inputenc}

\usepackage{hyperref}

\usepackage{fancyvrb}

\usepackage{enumitem}

 \newtheorem{thm}{Theorem}[section]
 
 \newtheorem{lemma}[thm]{Lemma}
 \newtheorem{prop}[thm]{Proposition}
 \theoremstyle{definition}

 \newtheorem{assumption}[thm]{Assumption}
 \numberwithin{equation}{section}
 
 \def\idtyty{{\mathchoice {\mathrm{1\mskip-4mu l}} {\mathrm{1\mskip-4mu l}} %
{\mathrm{1\mskip-4.5mu l}} {\mathrm{1\mskip-5mu l}}}}

\newcommand{\caA}{{\mathcal A}}

\newcommand{\caF}{{\mathcal F}}

\newcommand{\caH}{{\mathcal H}}

\newcommand{\caO}{{\mathcal O}}
\newcommand{\caP}{{\mathcal P}}

\newcommand{\caU}{{\mathcal U}}

\newcommand{\bbE}{{\mathbb E}}

\newcommand{\bbN}{{\mathbb N}}

\newcommand{\bbR}{{\mathbb R}}

\newcommand{\bbZ}{{\mathbb Z}}

\newcommand{\iu}{\mathrm{i}}

\newcommand{\str}{^{*}}
\newcommand{\ep}[1]{\mathrm{e}^{#1}}

\newcommand{\Tr}{\mathrm{Tr}}

\newcommand{\Qh}{Q_{\eta}}

\newcommand{\vertiii}[1]{{\left\vert\kern-0.25ex\left\vert\kern-0.25ex\left\vert #1 
    \right\vert\kern-0.25ex\right\vert\kern-0.25ex\right\vert}}

\DeclareMathOperator\supp{supp}

\title{Exactness of linear response in the quantum Hall effect}

\author{Sven Bachmann}
\address{Department of Mathematics \\ University of British Columbia \\ Vancouver, BC V6T 1Z2 \\ Canada}
\email{sbach@math.ubc.ca}

\author{Wojciech De Roeck}
\address{ Instituut Theoretische Fysica, KULeuven  \\
3001 Leuven  \\ Belgium }
\email{wojciech.deroeck@kuleuven.be}

\author{Martin Fraas}
\address{Department of Mathematics \\ Virginia Tech \\ Blacksburg, VA 24061-0123 \\ USA}
\email{fraas@vt.edu}

\author{Markus Lange}
\address{Department of Mathematics \\ University of British Columbia \\ Vancouver, BC V6T 1Z2 \\ Canada}
\email{mlange@math.ubc.ca}

\date{\today}

\begin{document}

\begin{abstract}
In general, linear response theory expresses the relation between a driving and a physical system's response only to first order in perturbation theory. In the context of charge transport, this is the linear relation between current and electromotive force expressed in Ohm's law. We show here that, in the case of the quantum Hall effect, all higher order corrections vanish. We prove this in a fully interacting setting and without flux averaging. 
\end{abstract}

\maketitle
 

\section{Introduction}\label{sec:Intro}


Quantization of the Hall conductance in the bulk is a well understood phenomenon under a spectral gap assumption, both in the absence and presence of interactions between the electrons, see \cite{TKNN,Thouless85,AvronSeilerSimon83,AvronSeiler85,ASS90} and \cite{HastingsMichalakis,OurQHE,giuliani2017universality,MBIndex,RationalIndex}.

The Hall conductance is a linear response coefficient and it can be expressed by Kubo's formula. In the presence of a spectral gap, the validity of linear response is well established, both in non-interacting~\cite{KuboBook} and in interacting~\cite{BDF16,MathAdiabatic,Teufel17,TeufelNEASS} situations. As was first pointed by Laughlin, the quantum Hall effect has a natural interpretation as a charge pump. In a cylindrical geometry where the current along the cylinder is induced by a time-dependent magnetic flux~\cite{Laughlin} (see also~\cite{Streda} for a similar idea), the quantum Hall conductance is directly proportional to the charge crossing a fiducial line winding around the cylinder as the flux is slowly increased by one quantum unit. This aspect of the Laughlin argument (the other aspect being the quantization itself) was generalized to the interacting context in~\cite{OshikawaConstraints,MBIndex}. 

Let $H_\phi$ be a smooth family of Hamiltonians parametrized by the magnetic flux $\phi$, and let $P_\phi$ by the corresponding ground state projections. Kato constructed in~\cite{Kato50} an `adiabatic' propagator $U_\mathrm{A}(\phi)$ such that
\begin{equation*}
P_\phi = U_\mathrm{A}(\phi) P_0 U_\mathrm{A}(\phi)\str.
\end{equation*}
If $\Delta Q_\mathrm{A}$ is the expectation value of the charge transported by $U_\mathrm{A} = U_\mathrm{A}(2\pi)$ across the fiducial line, then the Laughlin argument concludes, formally, that
\begin{equation*}
\Delta Q_\mathrm{A} = 2\pi \sigma_\mathrm{H} \in\bbZ,
\end{equation*}
where $\sigma_\mathrm{H}$ is the Hall conductance{\footnote{We work in units where the flux quantum is $2\pi$, the electric charge is $-1$ and $\hbar = 1$.}}.

In experiments, the quantization is not exact, but the conductance is quantized to nearly one part in a billion \cite{QHEwiki}.  That suggests that the universality expressed by the quantization of conductance extends to the charge transport of the full driven Schr\"odinger equation. If $\phi$ is changing slowly and smoothly in time, namely $\phi = \phi(\epsilon t)$ with $\epsilon\ll 1$, one may consider the charge transported by the physical propagator $U_\epsilon(s)$ associated with the time-dependent Hamiltonian $H_{\phi(s)}$ in rescaled time $s = \epsilon t$. The adiabatic theorem in the presence of a gap ensures that $U_\mathrm{A}(\phi(s))$ approximates $U_\epsilon(s)$ as $\epsilon\to 0$. And indeed, \cite{KleinSeiler} shows that power-law corrections to Kubo's formula for the flux averaged Hall conductance vanish to all orders in $\epsilon$.

In this work, we prove the result in a fully interacting setting without averaging. Just as in~\cite{KleinSeiler}, the fundamental reason of this exactness is to be found in the adiabatic theorem: If the driving is smooth, then the Schr\"odinger and adiabatic flows are equal to all orders in adiabatic perturbation theory~\cite{Berry90,Nenciu,MathAdiabatic} as soon as the driving has stopped. While~\cite{KleinSeiler} is geometric in nature, the present interacting result uses the many-body index of~\cite{MBIndex} and relies on locality arguments to leverage on the adiabatic theorem. Accordingly, Kato's propagator $U_\mathrm{A}$ is replaced with Hastings' local propagator $U_{\parallel}$ introduced in~\cite{Hastings:2004go,HastingsWen,BMNS} whose action coincides with Kato's on the ground state space. The corresponding charge transport is denoted $\Delta Q_\parallel$.

Let us finally comment on the volume dependence of our results. We shall work here in an arbitrary large but finite volume, with errors vanishing faster than any inverse power of the diameter $L$ of the system. Concretely, that means on the one hand that we have $\Delta Q_\parallel = 2\pi \sigma_\mathrm{H} + \caO(L^{-\infty})$ (and it is an integer up to $\caO(L^{-\infty})$). On the other hand, if $\Delta Q_\epsilon$ is the charge transported by the physical $U_\epsilon = U_\epsilon(2\pi)$, our main result reads
\begin{equation*}
\Delta Q_\epsilon = 2\pi \sigma_\mathrm{H} + \caO(\epsilon^\infty) + \caO_\epsilon(L^{-\infty}),
\end{equation*}
see Theorem~\ref{thm:KuboExact}. Here, the bound $\caO(\epsilon^\infty)$ is meant to be uniform in $L$, whereas the notation $\caO_\epsilon(L^{-\infty})$ indicates that the bound holds only pointwise in $\epsilon$. From there, it is a separate question whether $\sigma_\mathrm{H}$ has a well-defined limit as $L\to\infty$. As argued in~\cite{OurQHE}, this can also be answered positively.


\section{The Laughlin pump as a many-body index}\label{sec:Laughlin}


The Laughlin argument is traditionally exposed in a cylindrical setting with boundaries connected to infinite reservoirs. We shall work here in a periodic setting, by glueing the ends of the cylinder into a 2-dimensional torus. Furthermore, we consider a quantum lattice system defined on the torus and work in a large but finite volume. For clarity of the presentation, we denote $\Gamma$ the set of vertices of the system and $L$ the diameter of~$\Gamma$, expressed by the graph distance on~$\Gamma$. In this expository section, we ignore errors that vanish fast as $L\to\infty$.

In this setting, the quantum Hall effect has a natural interpretation as a charge pump and we are interested in the charge transported across a fiducial line $\nu_-$ winding across the torus. The physical source of the pumping is a slowly increasing magnetic flux threading the system, see~Figure~\ref{fig:Laughlin}.

By charge, we mean here that there is a family of operators $q(Z)$ labelled 
by subsets {$Z\subset\Gamma$} that have diameter smaller than a fixed value $R_q$; Typically, $R_q=1$. Crucially, these operators have integer spectrum and they mutually commute. The charge in any set $X\subset\Gamma$ is $Q_X = \sum_{Z\subset X} q(Z)$ and it has integer spectrum.

The system's dynamics is described by a Hamiltonian $H = \sum_{Z\subset \Gamma}\Phi(Z)$, where the interactions $\Phi(Z) = \Phi(Z)\str$ are uniformly bounded and finite range, namely $\Phi(Z) = 0$ whenever the diameter of $Z$ is larger than $R_\Phi$. The dynamics generated by $H$ conserves charge in the sense that $[H,Q_Z]$ is an operator supported along the boundary of $Z$. It follows that if $\Qh$ denotes the charge in a half system with one boundary $\eta_-$ (the other one being at the `other end of the universe' $\eta_+$, see Figure~\ref{fig:Laughlin}), then the gauge-transformed
\begin{equation}\label{Twist Antitwist}
\tilde H_\phi = \ep{\iu\phi \Qh}H\ep{-\iu\phi \Qh}
\end{equation}
differs from $H$ only along $\eta_\pm$ for all $\phi\in[0,2\pi]$. Finally, we consider the family $H_\phi$ obtained by gauge transforming only those interaction terms in the vicinity of~$\eta_-$, see Section~\ref{app:parallel} in the appendix for details. This is no longer a unitarily equivalent family and $H_\phi$ corresponds to having a flux $\phi$ threaded in the torus, see again Figure~\ref{fig:Laughlin}. Integrality of the spectrum of $\Qh$ implies that $H_{2\pi} = H_0 = H$.
\begin{figure}
\includegraphics[width=0.75\textwidth]{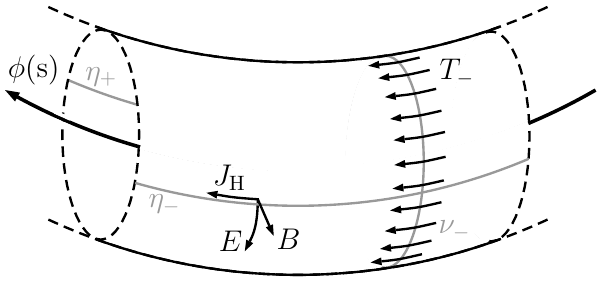}
\caption{The Laughlin pump: A slowly varying magnetic flux $\phi(s)$ threading the torus induces an electromotive force $E$ along the surface and, in the presence of a magnetic field $B$ piercing the surface, a Hall current $J_\mathrm{H}$ perpendicular to it. The expectation value of the charge transport $T_-$ across the fiducial line $\nu_-$ as $\phi$ increases by one flux quantum equals the Hall conductance.}
\label{fig:Laughlin}
\end{figure}

{ We assume that the spectral gap above lowest eigenvalue of $H_\phi$ is lower bounded by a positive constant $\gamma$, uniformly in $L$ and $\phi$, see Assumption~\ref{assum:gap} below. Let $P_\phi$ be the corresponding spectral projection and $p$ be its rank.}
The result of~\cite{Hastings:2004go,HastingsWen}, refined in~\cite{BMNS}, is the existence of an operator $ K_\phi$ supported only in a neighbourhood of the line $\eta_-$, which implements parallel transport on $P_\phi$. Indeed, the propagator $U_{\parallel}(\phi)$ defined by
\begin{equation}\label{Uparallel}
\partial_\phi U_{\parallel}(\phi) = \iu K_\phi U_{\parallel}(\phi),
\qquad U_{\parallel}(0) = \idtyty,
\end{equation}
satisfies
\begin{equation}\label{Parallel Transport}
P_\phi
= U_{\parallel}(\phi) P U_{\parallel}(\phi)\str
\end{equation}
for all $\phi\in[0,2\pi]$. Here we denoted $P = P_0$. In particular, 
\begin{enumerate}[label=(\alph*)]
\item Since the family of Hamiltonians is periodic, $P = P_{2\pi}$ and so
\begin{equation*}
[U_{\parallel}, P] = 0,
\end{equation*}
where $U_{\parallel} = U_{\parallel}(2\pi)$.
\item Since $K_{\phi}$ is supported in a neighbourhood of $\eta_-$, the unitary $U_{\parallel}(\phi)$ acts non-trivially only in that same neighbourhood. See again Appendix~\ref{app:parallel}.
\end{enumerate}

In the quantum Hall effect, the Kubo formula of linear response is equal to the adiabatic curvature, as was first observed in~\cite{AvronSeiler85}. In fact, it can be further related directly to the parallel transport corresponding to the addition of a flux quantum. Let $Q$ be the charge on the `orthogonal' half torus having $\nu_-$ as one of its boundaries, see Figure~\ref{fig:Laughlin}. The operator of charge transported by $U_{\parallel}$ (over a full cycle increasing the flux by $2\pi$) denoted $T_{\parallel} = U_{\parallel}\str Q U_{\parallel} - Q$ is a sum of two contributions $T_{\parallel,\pm}$ supported along $\nu_\pm$. Global charge conservation implies that its expectation value in the invariant state $P$ vanishes, but in general, this is only due to a cancellation of the influx of charge at one boundary and the outflux at the other one. Focussing on just one of them, it is proved in~\cite{MBIndex}{, Theorem~3.2} that the expected charge transport is equal to the Hall conductance
\begin{equation}\label{Kubo equals transport}
\Delta Q_\parallel = { p^{-1}}\Tr(PT_{\parallel,-}) = 2\pi \sigma_\mathrm{H},
\end{equation}
a fact at the heart of the original Laughlin argument. { It is worth pointing out here that although $\Delta Q_\parallel$ is the integrated current, it is equal to the Hall conductance $\sigma_\mathrm{H}$ of the system at $\phi = 0$, not its flux average.} 


\section{Charge transport and the adiabatic theorem}\label{sec:Adiabatic}


The propagator $U_{\parallel}$, featuring in Laughlin's argument,  is an approximation of the `true' adiabatic evolution of the system. By this, we mean the solution of the driven Schr\"odinger equation for the time dependent family of Hamiltonians $H_{\phi(t)}$ for a slowly varying flux $\phi(t)$. In rescaled time $s = \epsilon t \in [0,1]$, the Schr\"odinger propagator is the solution of
\begin{equation}\label{Ueps}
\iu\epsilon \partial_s U_\epsilon(s) = H_{\phi(s)}U_\epsilon(s),\qquad U_\epsilon(0) = \idtyty,
\end{equation}
and the adiabatic regime is characterized by $0< \epsilon\ll 1$.

The adiabatic theorem for gapped many-body systems ~\cite{MathAdiabatic,Teufel17} now goes as follows: For any local observable $A$, 
\begin{equation*}
\vert \Tr(P U_\epsilon(s)\str A U_\epsilon(s) ) - \Tr(P U_{\parallel}(s)\str A U_{\parallel}(s) )\vert
\leq C \Vert A\Vert \vert \mathrm{supp}(A) \vert^2 \epsilon\,,
\end{equation*}
where $U_{\parallel}(s) = U_{\parallel}(\phi(s))$ and $C$ is independent of the volume of the system. Furthermore, if the driving is smooth and has stopped at $s=1$, that is $\partial_s H_s$ is compactly supported in $(0,1)$, then the error is in fact beyond perturbation theory in the sense that
\begin{equation}\label{adiabatic thm}
\vert \Tr(P U_\epsilon\str A U_\epsilon) - \Tr(P U_{\parallel}\str A U_{\parallel})\vert
\leq C_m \Vert A\Vert \vert \mathrm{supp}(A) \vert^2 \epsilon^m
\end{equation}
for all $m\in\bbN$.

As in the previous section, we consider the operator of charge transported by $U_\epsilon$, denoted $T_\epsilon = U_\epsilon\str Q U_\epsilon - Q$ and its contribution $T_{\epsilon,-}$ across $\nu_-$. We define $\Delta Q_\epsilon = p^{-1}\Tr(PT_{\epsilon,-})$. While~(\ref{adiabatic thm}) make the following quite plausible, it is a non-trivial result that $\Delta Q_\epsilon = \Delta Q_\parallel + \caO(\epsilon^\infty)$. Combined with~(\ref{Kubo equals transport}), we therefore obtain that 
\begin{equation}\label{The result}
\Delta Q_\epsilon = {2\pi} \sigma_\mathrm{H} + \caO(\epsilon^\infty).
\end{equation}
{ We point out that a direct application of~(\ref{adiabatic thm}) would yield an error bound that diverges at any order in $\epsilon$ as $L\to\infty$. Indeed, $A$ would need to be the charge in a macroscopic region. Refining this naive estimate to obtain~(\ref{The result}) is the main technical aspect of our result.}

In physical terms, the driving is in this setting the electromotive force $E$ generated by the time dependent flux through Faraday's law:
\begin{equation*}
E = -\partial_t\phi = -\epsilon\phi'(s)
\end{equation*}
Hence, (\ref{The result}) can be written as $\Delta Q_\epsilon = {2\pi}\sigma_\mathrm{H} + \caO(\vert E\vert^\infty)$, expressing the exactness of linear response (namely, Ohm's law) to all orders {in the driving} for the quantum Hall effect. The fact that the response is, as in~\cite{KleinSeiler}, the integrated current is imposed by the present setting of a charge pump.

Let us briefly comment on some technicalities to come. The first class of difficulties arise from the fact that the charge operators $Q,\Qh$ have both norms and supports that grow with $L$. In view of the error term in~(\ref{adiabatic thm}), one may fear that the bound in~(\ref{The result}) cannot be uniform in the volume. A careful observation will however show that this is not the case, since the effective transport is limited by charge conservation to the vicinity of $\nu_-$, while the current is driven only along $\eta_-$ where the Hamiltonian itself changes (through the step in the gauge potential). It follows that $\sigma_\mathrm{H}$ and the operators of charge transport are localized in an $L$-independent region around the intersection of $\eta_-$ and $\nu_-$. The second difficulty is that while the Schr\"odinger and parallel transport flows agree on the ground state space to all orders in $\epsilon$ at $\phi = 2\pi$, the charge transport itself happens throughout the full driving, along which the error is only of order $\epsilon$. We bypass this issue by comparing $U_\epsilon(s) P U_\epsilon(s)\str$ only with a suitably dressed ground state projection, both remaining $\caO(\epsilon^\infty)$-close to each other for all $s$, while the dressed projection merges with the instantaneous ground state projection when the driving stops.

Finally, we compare the present work with~\cite{KleinSeiler}. The result itself differs in two respects. Firstly, the error bound we obtain is uniform in the volume, while that in~\cite{KleinSeiler} diverges in the number of particles, although that fact is not explicit for example in~Theorem~A4 of~\cite{KleinSeiler}. Secondly, the conductance there is defined through an average over the flux torus, just as in the original work~\cite{AvronSeiler85}. The methods differ significantly, too. As already pointed out, in order to obtain volume independent bounds, we use a local parallel transport instead of the traditional one of Kato. We further bypass the geometric argument, in particular the Chern-Simons formula and the need for averaging, by using the many-body index of~\cite{RationalIndex}. 


\section{Equality of charge transports}\label{sec:Core}


\subsection{Spatial setup}

We consider a quantum lattice system defined on a large but finite two-dimensional torus. Let $L\in \bbN$ be even and let $\Gamma = \Gamma_L= \bbZ_L^2$ where $\bbZ_L = \bbZ / (L\bbZ)$ is identified with $\{-L/2 +1,\ldots,L/2 - 1, L/2\}$. Note that $|\Gamma| = L^2$ and that $\mathrm{diam}(\Gamma) = L$ in the graph distance $\mathrm{d}$. We denote by $\eta$ the `horizontal' strip 
$\{(x_1,x_2)\in\Gamma: 0 \leq x_2 < L/2\}$ with boundary $\eta_{-} := \{(x_1,x_2)\in\Gamma: x_2 = 0\}$ and $\eta_{+} := \{(x_1,x_2)\in\Gamma: x_2 = L/2-1\}$. 
We similarly denote the `vertical' strip $\nu := \{(x_1,x_2)\in\Gamma: 0 \leq x_1 < L/2\}$ and its boundaries $\nu_{-} := \{(x_1,x_2)\in\Gamma: x_1 = 0\}$ and $\nu_{+} := \{(x_1,x_2)\in\Gamma: x_1 = L/2-1\}$.

Let $\caA$ be the even subalgebra of the CAR algebra generated by $\{1,a_x,a_x\str:x\in\Gamma\}$, which we can think of as acting on the antisymmetric Fock space
\begin{equation*}
\caH=\caH_L = \caF_{\mathrm{a}}(l^2(\Gamma)).
\end{equation*}
An observable $O\in \caA$ is said to be supported in $\Lambda\subset \Gamma$ if $O$ can be expressed as an even polynomial in $\{1,a_x,a_x\str:x\in\Lambda\}$, and we denote by $\supp(O)$ the smallest set on which $O$ is supported. Crucially, if $\supp(O_X) \subset X$ and $\supp(O_Y)\subset Y$ and $X, Y \subset \Gamma$ are disjoint, then $[O_X, O_Y] = 0$. Note that for all of what follows, $\caA$ could equivalently be taken to be the matrix algebra of a finite quantum spin system defined on~$\Gamma$. 

For our current purposes, it is handy to consider a weaker notion of support.    
 We say that an operator $O$ is almost supported in a set $Z$ if there is a sequence $(O_r)$, $r \in \bbN$ with $\supp(O_{r}) \subset Z_{(r)}$ such that 
\begin{equation}\label{Almost localized}
	\|O -O_{r}\| = \|O\| |Z| \caO(r^{-\infty}),
\end{equation}
where $Z_{(r)}$ denotes the $r$-fattening of $Z$, namely
\begin{align*}
	Z_{(r)} := \{x \in \Gamma \,:\, {\rm dist}(x,Z) \leq r\} \,.
\end{align*}
The class of such operators is denoted by $\caA_Z$.  
We point out that the constants implicit in the notation $ \caO(r^{-\infty}) $ of~(\ref{Almost localized}) do not depend on {the size $L$ of the system nor the adiabatic parameter $\epsilon$.  In fact, the notion of almost support make sense only for a family of operators labeled by $L$ and $\epsilon$, almost supported in a family of sets labeled by $L$ and $\epsilon$.

We believe that the chance of misunderstanding coming from this convention is small, but we briefly provide more details for the sake of completeness. Let $O_{L,\epsilon}$ be a family of operators defined on  a sequence of operator algebras $\mathcal{A}=\mathcal{A}_{L}$ associated with tori ${\Gamma}_L$ with diameter $L$, and let $Z_{L,\epsilon}$ be a family of subsets of ${\Gamma}_L$.
  Neither the sets $Z_{L,\epsilon}$ nor the operators $O_{L,\epsilon}$ for different values of $L,\epsilon$ are a priori related in any way.  The family $(O_{L,\epsilon})$ belongs to the set $\mathcal{A}_{(Z_{L,\epsilon})}$ if for all $r$ there exists a sequence of operators $O_{L,\epsilon,r}$ supported in the fattening $(Z_{L,\epsilon})_{(r)}$ such that for any $k\in\bbN$, there are constants $C_k$ such that
$$
\|O_{L,\epsilon} - O_{L,\epsilon,r} \| \leq  C_k \|O_{L,\epsilon}\| |Z_{L,\epsilon}| r^{-k},
$$ 
for all $L$ and $\epsilon$. See the appendix of~\cite{RationalIndex} for an extended discussion. In the remaining text we omit the index $L$, but keep track of $\epsilon$.}

We note that (\ref{Almost localized}) implies that 
\begin{equation*}
\Vert [O_X,O_Y] \Vert = \Vert O_X\Vert \Vert O_Y\Vert \vert X\vert\vert Y\vert \caO(d(X,Y)^{-\infty})
\end{equation*}
whenever $O_X\in\caA_X$ and $O_Y\in\caA_Y$.
%

\subsection{Extensive observables and the Lieb-Robinson bound} \label{subsec: extensive}

An extensive observable
\begin{align}\label{potential}
	S =  \sum_{X \subset \Gamma} \Psi(X)
\end{align} 
is a sum of local terms $\Psi(X) = \Psi(X)\str\in\caA$ with $\supp(\Psi(X)) = X$ that  satisfy 
\begin{enumerate}
	\item finite range condition: There is $R_\Psi < \infty$ such that $\Psi(X) = 0$ if ${\rm diam}(X) > R_\Psi$ ,
	\item finite interaction strength: There is $m_\Psi < \infty$ such that $\|\Psi(X)\| \leq m_\Psi$ for all $X$. 
\end{enumerate}
A few remarks. First of all, the above constants are understood to be independent of the system size $L$. Second of all, while the decomposition~(\ref{potential}) is not unique, we make it part of the definition\footnote{Strictly speaking, `extended observables' are hence not operators but functions $\caP(\Gamma) \to \caA:  X \mapsto \Psi(X)$.}.
 Thus, there is a natural restriction to a subset $Z\subset\Gamma$ given by $S_Z = \sum_{X \subset Z} \Psi(X)$. Finally (i,ii) imply immediately that $\Vert S_Z\Vert \leq C m_\Psi (R_\Psi)^2 \vert Z\vert$ in the present two-dimensional setting.

The first extensive observable of interest is the charge, which is denoted 
\begin{align*}	
	Q_\Gamma = \sum_{X \subset \Gamma} q(X)\,,
\end{align*}
where the local charge operator $q(X)$ is zero if $\mathrm{diam}(X) > R_q$ and satisfy
\begin{equation*}
\mathrm{Spec}(q(X))\subset\bbZ,\qquad \text{and}\qquad
[q(X),q(X')] = 0\textrm{ for all }X, X' \,.
\end{equation*}
A natural choice in the present fermionic setting is $q(X) \neq 0$ if and only if $X$ is a singleton in which case $q(\{x\}) = a\str_x a_x$. In a quantum spin system, $q(X)$ could be e.g.\ plaquette of vertex operators.

An extensive observable $S =  \sum_{X \subset \Gamma} \Psi(X)$ is called charge conserving if 
\begin{equation}\label{LCC}
[\Psi(X), Q_\Gamma] = 0.
\end{equation}
Since $[\Psi(X),Q_Y] =0$ for disjoint sets $X, Y$, we see that $[\Psi(X), Q_Z] = [\Psi(X), Q_\Gamma] = 0$ whenever $\mathrm{d}(X,Z^c) > R_q$. Hence, for any $Z\subset\Gamma$ and defining $\partial Z= \{\mathrm{d}(\cdot,Z)\leq 1\} \cap 
\{\mathrm{d}(\cdot,Z^c)\leq 1\}$, 
\begin{equation*}
[S,Q_Z] = \sum_{\substack{X\subset\Gamma: X\cap Z\neq \emptyset \\ \mathrm{d}(X,Z^c) \leq R_q}}[\Psi(X),Q_Z]\in \caA_{\partial Z}
\end{equation*}
for any charge conserving extensive observable $S$, a property that will play an important role in the following. 

The second extensive observable we introduce is a periodic family of charge conserving Hamiltonians, parametrized by $s\in [0,1]$, namely
\begin{align*}
	H_s = \sum_{X \subset \Gamma} \Phi_s(X),\quad\text{such that}\quad H_0 = H_1.
\end{align*}
In this case, the constants appearing in (i,ii) above are assumed to be independent of the parameter $s$. By~(\ref{LCC}),
\begin{equation*}
[H_s,Q_\eta] \in \caA_{\eta_-} + \caA_{\eta_+},\qquad 
[H_s,Q_\nu] \in \caA_{\nu_-} + \caA_{\nu_+}.
\end{equation*}
We write
\begin{equation}\label{H Charge cons}
\iu [H_s,Q_\nu] = \sum_{X \subset \Gamma} (j_{s,{\nu_-}}(X) + j_{s,{\nu_+}}(X))
\end{equation}
and note that $j_{s, \nu_{\pm}}(X)$ is zero if $\mathrm{diam}(X) > R_q + R_\Phi$ or $X\cap\nu_\pm = \emptyset$.

The Lieb-Robinson bound for a dynamics $\tau_t(O_X) = \ep{\iu t S} O_{X}\ep{-\iu t S}$ on $\Gamma$ generated by an extensive observable $S$ implies that if $\mathrm{supp}(O_X) = X$ then for any $\delta>0$, 
\begin{equation}\label{LR bound}
\Vert \tau_t(O_X) - \bbE_{X_{(vt+\delta)}}(\tau_t(O_X))\Vert
\leq C\Vert O_X \Vert\vert X \vert \ep{-\zeta \delta},
\end{equation}
where $\bbE_Z$ is the partial trace over the complement of the set $Z$ and the positive constants $C,v,\zeta$ depend on $S$ but not on {$t$}, { see~\cite{LRYoshiko} and }Section~\ref{Appendix LR} in the appendix. 
The Lieb-Robinson bound continues to hold in the case of time-dependent generators. This applies in particular for the dynamics implemented by the unitary $U_\epsilon$ introduced in~(\ref{Ueps}) and corresponding to a time of order $\epsilon^{-1}$. We then find that $U_\epsilon^* O_X U_\epsilon \in \caA_{X_{(v\epsilon^{-1})}}$ whenever $O_X\in\caA_X$.

\subsection{Assumptions} 

We can now state our assumptions for the following results. As in the rest of this section, we keep the setting more general than the specific Hall cylinder described in the introductory sections, where charge is being pumped by an increase of the magnetic flux. Still, the assumptions below, in particular Assumption~\ref{assum:LocalDriving}, are made with this example in mind. The charge transported by parallel transport is not, in general, related to a linear response coefficient, which is why the main theorem does not explicitly refer to a conductance.

Let $P_s$ denote the orthogonal projection onto the ground state space of $H_s$. Let $E^0_s,E^1_s$ be the ground state energy and the energy of the first excited state. $P_s$ is necessarily finite dimensional, but its rank may grow in the system size $L$. Our first assumption is that this is not the case. The second part below is the crucial gap assumption.

\begin{assumption}[Spectral gap] \label{assum:gap}
There are constants $\gamma > 0, p \in \bbN$ and an $L_0\in\bbN$ such that if $L>L_0$, then \\
(i) $\textrm{rk}(P_s) = p$ for all $s\in[0,1]$ and \\
(ii) $ E^1_s - E^0_s \geq \gamma$ for all $s\in[0,1]$.
\end{assumption}

The second assumption is about smoothness of the function $s\mapsto H_s$ and the fact that the driving is compactly supported in $(0,1)$.

\begin{assumption}[Smooth, compactly supported driving] \label{assum:LocAndDif}
The matrix-valued function $s\mapsto \Phi_{s}(X)$ is infinitely often continuously differentiable for any $X \subset \Gamma$ and $\Phi^{(k)}_{s=0}(X) = \Phi^{(k)}_{s=1}(X) = 0$ for all $k\in\bbN$. Moreover, these derivatives define extensive observables $H^{(k)}_s = \sum_{X\subset\Gamma}\Phi^{(k)}_{s}(X)$ in the sense of Section~\ref{subsec: extensive}. 
\end{assumption}

As already pointed out, our last assumption reflects the specific geometry of the quantum Hall cylinder depicted in Figure~\ref{fig:Laughlin}.

\begin{assumption}[Localized driving] \label{assum:LocalDriving}
The driving extends only along the line~$\eta_-$, namely { $\partial_s H_s \in \caA_{\eta_{-}}$.}
\end{assumption}
\noindent In the case of the Hall effect, this means that the gauge potential describing the threaded flux is constant in space and time away from the line $\eta_-$ and its change, which drives the Hall current, is spatially localized along that same line, see the discussion in Section~\ref{sec:Laughlin}. 

Concretely, the last two assumptions are satisfied for the family $H_s = H_{\phi(s)}$ introduced in Section~\ref{sec:Laughlin} and Appendix~\ref{app:parallel}, where $\phi\in C^\infty([0,1];\bbR)$ with $\phi(0)= 0, \phi(1) = 2\pi$ and $\phi'\geq 0$ is compactly supported in $(0,1)$. While Assumption~\ref{assum:gap} is believed to hold for (possibly fractional) quantum Hall systems, there is at the moment of writing no explicit microscopic model where this can be proved, except for perturbations of free systems (and therefore integer conductance), see~\cite{giuliani2017universality, DeRoeck2018, HastingsPerturbation}; see however~\cite{GappedHall}.

\subsection{Conductance, parallel transport and adiabatic evolution} 

Let us first recall the two main players transporting charge. On the one hand, $U_\epsilon(s)$ is the Schr\"odinger propagator for the slowly driven system with Hamiltonian $\epsilon^{-1} H_s$, see~(\ref{Ueps}). On the other hand $U_{\parallel}(s)$ implements parallel transport of $P_{\phi(s)}$, namely
\begin{equation}\label{Parallel transport}
P_{\phi(s)} = U_{\parallel}(s) P U_{\parallel}(s)\str. 
\end{equation}
It is generated by
\begin{equation}\label{Ks}
K_s = \int W(u) \ep{\iu u H_s} \partial_s H_s \ep{-\iu u H_s} du 
\end{equation}
where $\vert W(u)\vert = \caO(\vert u \vert^{-\infty})$. Assumption~\ref{assum:LocalDriving} together with the fast decay of $W$ and the Lieb-Robinson bound for $\ep{-\iu u H_s}$ imply that $K_s\in\caA_{\eta_-}$. See Section~\ref{app:parallel} and~\cite{MathAdiabatic} for more details. To compare with~(\ref{Uparallel}), $K_s = \phi'(s)  K_{\phi(s)}$.

While $K_s$ arises from the extensive observable $\partial_s H_s$, it does not satisfy the finite range condition~(i) in Section~\ref{subsec: extensive}. However, the Lieb-Robinson bound for $\ep{-\iu u H_s}$ and the fast decay of $\vert W\vert$ imply that $K_s = \sum_{X\subset\Gamma}k_s(X)$ with
\begin{equation}\label{Extensive}
\vertiii{K}_f :=\sup_s\sup_{x,y\in\Gamma}\sum_{X\subset\Gamma, X\ni x,y}\frac{\Vert k_s(X)\Vert}{f(\mathrm{d}(x,y))}<\infty,
\end{equation}
uniformly in $L$, for a positive, decreasing function $f$ such that $f(r) = \caO(r^{-\infty})$. See~\cite{BMNS} for details. We refer to operators with such finite $\vertiii{\cdot}_f$ as generalized extensive observables. The finiteness of~$\vertiii{\cdot}_f$ implies that 
\begin{equation*}
\Vert [K, O_X]\Vert\leq C \Vert O_X\Vert\vert X\vert
\end{equation*}
where $X = \mathrm{supp}(O_X)$.

The subexponential decay of $f$ carries over to the Lieb-Robinson bound and~(\ref{LR bound}) for the dynamics $\sigma_t$ generated by such a generalized extensive observable must be weakened to the following: For any $\alpha>1$ and any $\mu>0$, there are $c,C>0$ such that if $\mathrm{supp}(O_X) = X$ and for all $\delta>0$,
\begin{equation}\label{Slow LR}
\Vert \sigma_t(O_X) - \bbE_{X_{(ct^\alpha + \delta)}}(\sigma_t(O_X))\Vert
\leq C \Vert O_X\Vert\vert X \vert \ep{-\mu \delta^{1/\alpha}},
\end{equation}
see Section~\ref{Appendix LR} in the appendix. This difference is irrelevant for times of order $1$, but it is essential at the adiabatic time scale $\epsilon^{-1}$. Since the exact value of $\alpha$ will bear no effect on the final result, we chose $\alpha = 2$ for the rest of this paper. With this choice, (\ref{Slow LR}) implies that $\sigma_{\epsilon^{-1}}(O_X)  \in \caA_{X_{(c\epsilon^{-2})}}$.

Recall the notation $Q = Q_\nu$ used in the introductory sections. By charge conservation~(\ref{LCC}),
\begin{equation}\label{K Charge conservation}
[ K_s, Q] \in \caA_{\nu_-} + \caA_{\nu_+},
\end{equation}
and the Lieb-Robinson bound~(\ref{Slow LR}) imply that for $U_\parallel = U_\parallel(1)$
\begin{equation}\label{Tparallel}
U_\parallel\str Q U_\parallel - Q = \iu\int_0^1 U_\parallel(s)\str[ K_s, Q] U_\parallel(s) ds
= T_{\parallel,-} + T_{\parallel,+}
\end{equation}
where $T_{\parallel,-}\in\caA_{\nu_-}$ and $T_{\parallel,+}\in\caA_{\nu_+}$. We immediately note, and shall use it later, that $K_s\in\caA_{\eta_-}$ implies that $T_{\parallel,\pm}\in\caA_{\nu_\pm\cap\eta_-}$.

The same holds for $U_\epsilon = U_\epsilon(1)$, in the sense that
\begin{align}
U_\epsilon\str Q U_\epsilon - Q 
&=\frac{\iu}{\epsilon}\int_0^1 U_\epsilon\str(s) [H_s, Q] U_\epsilon(s) ds \nonumber \\
&=\frac{1}{\epsilon}\int_0^1 U_\epsilon\str(s) \Big(\sum_{X \subset \Gamma} (j_{s,{\nu_-}}(X) + j_{s,{\nu_+}}(X))\Big) U_\epsilon(s) ds \nonumber \\
&= T_{\epsilon,-} + T_{\epsilon,+},\label{Teps}
\end{align}
see~(\ref{H Charge cons}). Since the time evolution runs over a long time $\epsilon^{-1}$, the transport observables $T_{\epsilon,\pm}$ are almost localized in a fattening of order $\epsilon^{-1}$ of $\nu_\pm$, see~(\ref{LR bound}).

To proceed, we fix a convention about error terms. Our statements depend on two asymptotic parameters $\epsilon,L$, but their roles are not symmetric. When writing $\mathcal{O}(\epsilon^\infty)$, we mean that the bound is uniform in $L$, whereas we write $\mathcal{O}_\epsilon(L^{-\infty})$ for a bound that is pointwise in $\epsilon$, but that does not necessarily hold uniformly.  Indeed, as is physically most relevant, we prefer to think of $\epsilon$ as an arbitrarily small but fixed parameter and let $L\to\infty$ first.
Then,  the main result of our work is
\begin{thm}\label{thm:KuboExact}
Let Assumptions~(\ref{assum:gap},\ref{assum:LocAndDif},\ref{assum:LocalDriving}) hold. Then
\begin{equation*}
\Tr(PT_{\parallel,-}) = \Tr(PT_{\epsilon,-}) + \mathcal{O}(\epsilon^\infty) +\mathcal{O}_\epsilon(L^{-\infty}).
\end{equation*}
\end{thm}
The theorem expresses in general the equality of two charge transports, independently of the fact that the left hand side is a linear response coefficient. In the more specific case of the Laughlin setting described in Section~\ref{sec:Laughlin} {where $p^{-1}\Tr(PT_{\parallel,-}) = \Delta Q_\parallel$ and $p^{-1}\Tr(PT_{\epsilon,-}) = \Delta Q_\epsilon$,} then the left hand side is, up to a factor $2\pi p$, the Hall conductance
\begin{equation}\label{parallel yields Kubo}
\Tr(P T_{\parallel,-}) = p \, 2\pi \sigma_\mathrm{H} + \caO(L^{-\infty}),
\end{equation}
and it is an integer, see~\cite{MBIndex}. In that case, the theorem states, as announced, that this linear response coefficient expresses the full charge transport, to all orders in the adiabatic parameter~$\epsilon$, and equivalently to all orders in the driving. 


\section{Dressed ground states and proofs}\label{sec:Proofs}

While the adiabatic theorem briefly discussed in the previous sections is the fundamental reason for the validity of the theorem, it will not appear in the proofs below per se. In fact, we shall in the following rather revisit the derivation of the adiabatic theorem presented in~\cite{MathAdiabatic}, with an additional fact that the driving is supported only along the line $\eta_-$. In this section, we shall not repeat the running assumptions of Theorem~\ref{thm:KuboExact}. We will freely use notations introduced in~\cite{MathAdiabatic}. 

\subsection{Dressing the ground state projection} 

 A key player in the proofs are the dressed ground state projections $\Pi_{n,\epsilon}(s),\, (n\in\bbN)$. On the one hand, they follow the driven projection $U_\epsilon(s) P U_\epsilon(s)\str$ to order $\caO(\epsilon^{n-2})$ for all $s\in[0,1]$. On the other hand they follow
the instantaneous (namely, parallel transported) ground state $P_s$ only to order $\caO(\epsilon)$ for all $s\in[0,1]$, but $\Pi_{n,\epsilon}(1) = P_1 = P$ exactly when the driving has stopped, see Assumption~\ref{assum:LocAndDif}.

The construction of $\Pi_{n,\epsilon}(s)$ is given in the proof of Lemma~4.3 in~\cite{MathAdiabatic}. It amounts to finding generalized extensive observables $\{A_j(s):j\in\bbN\}$ such that, with  $S_{n,\epsilon}(s) = \sum_{j=1}^n\epsilon^j A_j(s)$,
\begin{equation}\label{Dressing}
\Pi_{n,\epsilon}(s) = \ep{\iu S_{n,\epsilon}(s)} P_s \ep{-\iu S_{n,\epsilon}(s)}.
\end{equation}
Moreover, $A_j(s)$ are functions of $H_s$ and its derivatives, all at the same epoch $s$. On the one hand, this implies that $A_j(s)$ are charge conserving. On the other hand, whenever the derivatives vanish, so do the $A_j$, so that $A_j(0) = 0$ and $A_j(1) = 0$. Hence,
\begin{equation}\label{Pi(1)}
\Pi_{n,\epsilon}(0) = P = \Pi_{n,\epsilon}(1).
\end{equation}
On the other hand, there is a unitary propagator ${V}_{n,\epsilon}(s)$ such that
\begin{equation}\label{V}
\Pi_{n,\epsilon}(s) = {V}_{n,\epsilon}(s) P {V}_{n,\epsilon}(s)\str.
\end{equation}
To emphasize the difference with $\ep{\iu S_{n,\epsilon}(s)}$, we note that the right hand side contains $P$ and not $P_s$. Importantly, ${V}_{n,\epsilon}(s)$ is obtained as the solution of 
\begin{equation*}
\iu\epsilon \partial_s{V}_{n,\epsilon}(s) = (H_s + R_{n,\epsilon}(s)){V}_{n,\epsilon}(s),\qquad {V}_{n,\epsilon}(0) = \idtyty,
\end{equation*}
for a generalized extensive observable $R_{n,\epsilon}(s)$ which is small in the sense that
\begin{equation}\label{local norm of R}
\vertiii{R_{n,\epsilon}(s)}_f  \leq C\epsilon^{n+1},
\end{equation}
for some $f(r)$ vanishing faster than any inverse power, see \eqref{Extensive}. Importantly, $R_{n,\epsilon}(s)$ is obtained from multicommutators of $H_s$ and its derivatives $\{H_s^{(j)}:j=1,\ldots, n\}$ so that $R_{n,\epsilon}(s)\in\caA_{\eta_-}$ by Assumption~\ref{assum:LocalDriving}, and it is charge conserving. We refer again to the proof of Lemma~4.3 in~\cite{MathAdiabatic} for an explicit construction and proofs of the claimed properties.

We use $V_{n,\epsilon}$ as a link to compare the real evolution $U_{\epsilon}$ with the parallel transport $U_\parallel$. To that end we define 
\begin{equation}\label{def of W tilde}
	\widetilde{W}_{n,\epsilon}(s) = U_{\epsilon}(s)\str{V}_{n,\epsilon}(s)
\end{equation}
and 
\begin{equation}\label{def of W}
W_{n,\epsilon}(s) = V_{n,\epsilon}(s)\str \ep{\iu S_{n,\epsilon}(s)}U_\parallel(s).
\end{equation}

With Assumption~\ref{assum:LocAndDif} and the discussion above, $W_{n,\epsilon}(0) = \idtyty$, while at $s=1$,
\begin{equation}\label{W at 1}
\widetilde{W}_{n,\epsilon} = U_{\epsilon}\str{V}_{n,\epsilon},
\qquad W_{n,\epsilon} = V_{n,\epsilon}\str U_\parallel, 
\end{equation}
reduce to a comparison of the driven Schr\"odinger propagator with the implementation of parallel transport through an intermediary $V_{n, \epsilon}$. We note that
\begin{equation}\label{W tilde W}
W_{n,\epsilon}(s)\str P W_{n,\epsilon}(s) = P,
\end{equation}
see~(\ref{Parallel transport},\ref{Dressing},\ref{V}).

\subsection{Charge transport} \label{Sec: T-s}

We are finally equipped to prove the main theorem of this paper. The first lemma provides a bound on the charge transport operator associated with the auxiliary unitary $\widetilde{W}_{n,\epsilon}(s)\str$. It relies on the fact that the difference between the dressed unitary $V_{n,\epsilon}(s)$ and the Schr\"odinger propagator $U_\epsilon(s)$ is small in norm. This is also at the heart of the proof of the adiabatic theorem. Here, we shall moreover use the fact that the driving is localized along $\eta_-$, Assumption~\ref{assum:LocalDriving}.

We recall that the notation $\caO_\epsilon(L^{-\infty})$ refers to a bound that is not necessarily uniform in $\epsilon$.
Below, we will need the following structure and notation for an operator or a (generalized) extensive observable $S$. We write
$S=S_++S_-$ whenever $S_{\pm}$ have the property that 
$$
[S_{\pm}, A] = \caO_\epsilon(L^{-\infty}), \qquad \text{for any}\quad A\in \caA_{(\nu_{\mp})_{c\epsilon^{-2}}} , \Vert A \Vert=1.
$$
This also means that the splitting $S=S_++S_-$ is only defined up to errors of order $\caO_\epsilon(L^{-\infty})$. 
In fact, whenever we use this notation below, the objects $S_{\pm}$ will have much tighter localization properties than what follows from the above equation.  
We note that this notation was already used in \eqref{Teps}.

\begin{lemma}
\label{lem:BoundOnChargeTransport}
For all $s\in[0,1]$ and $n\in\bbN,n>9$,
\begin{equation*}
\Vert (\widetilde{W}_{n,\epsilon}(s) Q \widetilde{W}_{n,\epsilon}(s)^* - Q)_{-}\Vert = \mathcal{O}(\epsilon^{n-5}) + \caO_\epsilon(L^{-\infty}),
\end{equation*}
and
$$
\Vert \widetilde W_{n,\epsilon} (V_{n,\epsilon}\str Q V_{n,\epsilon} - Q)_- \widetilde W_{n,\epsilon}\str - (V_{n,\epsilon}\str Q V_{n,\epsilon} - Q)_- \Vert = \mathcal{O}(\epsilon^{n-9}) + \caO_\epsilon(L^{-\infty}).
$$
\end{lemma}

\begin{proof}
By its definition~(\ref{def of W tilde}), the unitary $\widetilde{W}_{n,\epsilon}(s)$ is the unique solution of the initial value problem
\begin{equation}\label{eq for W tilde}
	\iu\epsilon \partial_s{\widetilde{W}}_{n,\epsilon}(s) 
	=  U_{\epsilon}(s)^* R_{n,\epsilon}(s) U_{\epsilon}(s) \widetilde{W}_{n,\epsilon}(s) \,, \qquad 
	\widetilde{W}_{n,\epsilon}(0) = \idtyty.
\end{equation}
Hence for any operator $O$,
\begin{equation}
\widetilde{W}_{n,\epsilon}(s)\str O \widetilde{W}_{n,\epsilon}(s) - O
= \frac{\iu}{\epsilon}\int_0^s  V_{n,\epsilon}(r)\str [R_{n,\epsilon}(r) , U_{\epsilon}(r) OU_{\epsilon}(r)^*]V_{n,\epsilon}(r)  dr.
\label{Wtilde R Wtilde}
\end{equation}

As already pointed out, $R_{n,\epsilon}(r)$ is a generalized extensive observable conserving charge and $R_{n,\epsilon}(r)\in\caA_{\eta_-}$ by Assumption~\ref{assum:LocalDriving}. Now, 
\begin{equation*}
[R_{n,\epsilon}(r) , U_{\epsilon}(r)^* QU_{\epsilon}(r)]
= [R_{n,\epsilon}(r) , Q] + [R_{n,\epsilon}(r) , (U_{\epsilon}(r)^* QU_{\epsilon}(r) - Q)],
\end{equation*}
and the charge conservation equations (\ref{LCC}) and  (\ref{Teps}), imply that both commutators on the right-hand side can be decomposed as $S=S_-+S_+$, as announced above. By~(\ref{local norm of R}) and the fact that $R_{n,\epsilon}(r)\in\caA_{\eta_-}$, we find
\begin{equation*}
\Vert ([R_{n,\epsilon}(r) , Q])_-\Vert = \caO(\epsilon^{n+1}) + \caO_\epsilon(L^{-\infty}).
\end{equation*}
On the other hand, using (\ref{Teps}) we have 
$$
[R_{n,\epsilon}(r) , (U_{\epsilon}(r)^* QU_{\epsilon}(r) - Q)_-] = \sum_{X \subset \Gamma}\epsilon^{-1}\int_0^r [R_{n,\epsilon}(r),U_\epsilon\str(s) j_{s,{\nu_-}}(X)  U_\epsilon(s)] ds.
$$
Now, $R_{n,\epsilon}(r)$ is a generalised extensive observable almost supported in $\eta_-$, while $U_{\epsilon}(r)^* j_{s,{\nu_-}}(X)U_{\epsilon}(r) \in \caA_{X_{(v\epsilon^{-1})} \cap (\nu_{-})_{(v\epsilon^{-1})} } $  and it is zero if $\mathrm{diam}(X) > R_q + R_\Phi$, see~(\ref{H Charge cons}). Hence,
\begin{equation*}
\Vert [R_{n,\epsilon}(r) , (U_{\epsilon}(r)^* QU_{\epsilon}(r) - Q)_-]\Vert \leq  C \epsilon^{-3} \sup_{X \subset \Gamma} \|[R_{n,\epsilon}, U_\epsilon\str(s) j_{s,{\nu_-}}(X)  U_\epsilon(s)] \|.
\end{equation*}
We do not go through the details of this bound, but its validity follows roughly speaking from the fact that there are of order $\epsilon^{-2}$ terms with a leading order contribution to the commutator. With~(\ref{local norm of R}) and the fact that $U_\epsilon\str(s) j_{s,{\nu_-}}(X)  U_\epsilon(s)$ is almost supported on a set of volume of order $\epsilon^{-2}$, the commutator on the right hand side can be estimated by $C\epsilon^{n+1-2}$, so that
$$
\Vert [R_{n,\epsilon}(r) , (U_{\epsilon}(r)^* QU_{\epsilon}(r) - Q)_-]\Vert = \caO(\epsilon^{n-4}).
$$
Using $O = Q$ in (\ref{Wtilde R Wtilde}), we now conclude that
\begin{align*}
\Vert (\widetilde{W}_{n,\epsilon}(s)\str Q \widetilde{W}_{n,\epsilon}(s) - Q)_{-}\Vert
&\leq \epsilon^{-1} \sup_{r\in[0,s]}\Vert ([ R_{n,\epsilon}(r)  , U_{\epsilon}(r) Q U_{\epsilon}(r)^* ])_- \Vert \\
&= \caO(\epsilon^{n-5}) + \caO_\epsilon(L^{-\infty}).
\end{align*}
The first claim follows immediately from this estimate since $ (\widetilde{W}_{n,\epsilon}(s) Q \widetilde{W}_{n,\epsilon}(s)^* - Q)_{-} =  - \widetilde{W}_{n,\epsilon}(s)(\widetilde{W}_{n,\epsilon}(s)\str Q \widetilde{W}_{n,\epsilon}(s) - Q)_{-}\widetilde{W}_{n,\epsilon}(s)\str$.

A similar reasoning yields the second claim. The operator $O = V_{n, \epsilon}\str Q V_{n,\epsilon} - Q$ belongs to $\caA_{(\nu_-)_{(c\epsilon^{-2})}}$ and is a sum of terms that are almost localized on sets of diameter of order $\epsilon^{-2}$ { by~(\ref{Slow LR})}. But then {(\ref{LR bound}) further implies that} $U_{\epsilon}(r) O U_{\epsilon}(r)^* \in \caA_{(\nu_-)_{(c\epsilon^{-2})}}$ is, roughly speaking, a sum of local terms each of diameter $\epsilon^{-2}$. Hence, as above, there are order $\epsilon^{-4}$ terms contributing to the commutator in~(\ref{Wtilde R Wtilde}), each by an order $\epsilon^{-5}$, and (\ref{local norm of R}) yields
$$
\Vert [R_{n,\epsilon}(r) , U_{\epsilon}(r) O U_{\epsilon}(r)^*] \Vert = \caO(\epsilon^{n-8}) + \caO_\epsilon(L^{-\infty}).
$$
Using (\ref{Wtilde R Wtilde}) we get the last estimate of the lemma.
\end{proof}
We note that the lemma relies on the fact that both unitaries $U_\epsilon(s), V_{n,\epsilon}(s)$ being compared within $\widetilde{W}_{n,\epsilon}(s)$ are propagators over an adiabatic time $\epsilon^{-1}$, see~(\ref{eq for W tilde}). This would not be the case for the direct comparison of $U_\epsilon(s)$ with $\ep{\iu S_{n,\epsilon}(s)}U_\parallel(s)$ since the latter unitaries correspond to a time scale of order $1$.

We can now compare the charge transported by $U_\epsilon$ and by $V_{n, \epsilon}$.
\begin{lemma}\label{lem:TMinusAt1}
{For any $s\in[0,1]$, let $T_{\epsilon,-}(s)$} be as in~(\ref{Teps}).  For $n >9$,
\begin{equation} \label{eq:T-=TH-Tepsilon} 
\Tr(P T_{\epsilon,-}(s))
= \Tr(P (V_{n,\epsilon}(s)\str Q V_{n,\epsilon}(s) - Q)_-)  + \caO(\epsilon^{n-9}) + \caO_\epsilon(L^{-\infty}).
\end{equation}
\end{lemma}
\begin{proof}
All unitaries introduced so far are functions of the Hamiltonian and its derivatives. Therefore, they are all charge conserving in the sense of~(\ref{Tparallel},\ref{Teps}). It follows that
\begin{align*}
T_{\epsilon,-}(s) &= (\widetilde W_{n,\epsilon}(s)V_{n,\epsilon}(s)\str Q V_{n,\epsilon}(s)\widetilde W_{n,\epsilon}(s)\str - Q)_- \\
&= (V_{n,\epsilon}(s)\str Q V_{n,\epsilon}(s) - Q)_- \\
&\quad + \widetilde W_{n,\epsilon}(s) (V_{n,\epsilon}(s)\str Q V_{n,\epsilon}(s) - Q)_- \widetilde W_{n,\epsilon}(s)\str - (V_{n,\epsilon}(s)\str Q V_{n,\epsilon}(s) - Q)_- \\
&\quad  + (\widetilde W_{n,\epsilon}(s) Q \widetilde W_{n,\epsilon}(s)\str - Q)_- + \caO_\epsilon(L^{-\infty}).
\end{align*}
By Lemma~\ref{lem:BoundOnChargeTransport} we then get 
$$
\Vert T_{\epsilon,-}(s) - (V_{n,\epsilon}(s)\str Q V_{n,\epsilon}(s) - Q)_-  \Vert =  \mathcal{O}(\epsilon^{n-9}) + \caO_\epsilon(L^{-\infty}).
$$
which immediately yields the claim of the lemma.
\end{proof}

We now compare the expectation values of charge transported by $V_{n,\epsilon}$ and $U_\parallel$. Here, we rely on 
the many-body index
\begin{equation}\label{def index}
\mathrm{Ind}_P(U) = \Tr(P(U\str Q U - Q)_-)
\end{equation}
of~\cite{RationalIndex}. It allows to bypass the geometric picture of~\cite{KleinSeiler}. The index is defined for any locality and charge conserving unitary that commutes with $P$. The index belongs to $\mathbb{Z}$ up to an error of order $L^{-\infty}$. The error bound depends in particular on the unitary $U$. The fact that we apply the index theorem below for the $\epsilon$-dependent ${W}_{n,\epsilon}(s)$ (and again later in the proof of Theorem~\ref{thm:KuboExact} for $V_{n,\epsilon}$) is the fundamental origin of the non-uniformity in $\epsilon$ of the error vanishing as $L^{-\infty}$.

\begin{lemma}\label{cor:IndexIs0}
$\mathrm{Ind}_P({W}_{n,\epsilon}(s)) = \mathcal{O}_\epsilon(L^{-\infty})$ for all $s \in [0,1]$. 
\end{lemma}

\begin{proof}
We first show that $\mathrm{Ind}_P({W}_{n,\epsilon}(s))$ is well-defined for any $s\in[0,1]$. The unitary ${W}_{n,\epsilon}(s)$ satisfies a Lieb-Robinson bounds and conserves charge. The observation~(\ref{W tilde W}) can be written as
\begin{equation*}
[{W}_{n,\epsilon}(s), P] = 0
\end{equation*}
exactly, i.e.\ without any error in either $\epsilon$ or $L^{ -1}$. 
Hence, for any fixed $\epsilon$,  all assumptions of the index theorem in~\cite{RationalIndex} hold, yielding an integer-valued index associated with $P$ and ${W}_{n,\epsilon}(s)$.

The map $s \mapsto {W}_{n,\epsilon}(s)$ being differentiable, it is a fortiori continuous. It follows that $s\mapsto \mathrm{Ind}_P({W}_{n,\epsilon}(s))$ is constant up to $\caO_\epsilon(L^{-\infty})$, see~\cite[Proposition 2.2]{MBIndex}. The statement then follows from ${W}_{n,\epsilon}(0) = \idtyty$ and $\mathrm{Ind}_P(\idtyty) = 0$.
\end{proof}

We are now equipped to prove Theorem~\ref{thm:KuboExact}.
\begin{proof}[Proof of Theorem~\ref{thm:KuboExact}]
At $s=1$, we have $W_{n,\epsilon} = V_{n,\epsilon}\str U_\parallel$ and each unitary commutes individually with $P$, {because $P_1 = P_0 = P$ and by~(\ref{Pi(1)})}. Hence, each of them is associated with its own index. By additivity of the index, see~\cite[Proposition~5.1]{RationalIndex}, we then have
\begin{align*}
\mathrm{Ind}_P({W}_{n,\epsilon}) &= \mathrm{Ind}_P({V}_{n,\epsilon}^*) + \mathrm{Ind}_P({U}_\parallel) + \mathcal{O}_\epsilon(L^{-\infty}) \\
						    &= - \Tr(P (V_{n,\epsilon}\str Q V_{n,\epsilon} - Q)_-) + \Tr(P T_{\parallel,-}) + \mathcal{O}_\epsilon(L^{-\infty}).
\end{align*}
The statement follows by Lemma~\ref{cor:IndexIs0} and Lemma~\ref{lem:TMinusAt1} at $s=1$ since $n\in\bbN$ is arbitrary. 
\end{proof}



\section*{Acknowledgements}

\noindent We thank Stefan Teufel for bringing this problem to our attention. S.B.\ wishes to thank Marcel Schaub for many related discussions. The work of S.B.\ and M.L.\ was supported by NSERC of Canada. M.F. was supported in part by the NSF under grant DMS-1907435. W.D.R.\ thanks the Flemish Research Fund (FWO) for support via grants G076216N and 
G098919N.

\section{Appendix}\label{sec:Appendix}


\subsection{On parallel transport}\label{app:parallel}

The reader may have noticed that the formula~(\ref{parallel yields Kubo}) expressing the equality of conductance with an index appears in fact in a different fashion in the cited~\cite{MBIndex}. The proposition we wish to prove in this appendix shows that they are indeed the same.

In order to avoid confusion, we insist that all indices $(\cdots)_-$ appearing in this section refer to the boundary of $\eta$, not of $\nu$, see Figure~\ref{fig:Laughlin}.

First of all, recall the definition~(\ref{Twist Antitwist}) of the `twist-antitwist Hamiltonian'~$\tilde H_\phi$ which arises from $H$ by a gauge transformation. By contrast, the `twist Hamiltonian' $H_\phi = \sum_{X\subset\Gamma}\Phi_\phi(X)$ is obtained from $H = \sum_{X\subset\Gamma}\Phi(X)$ by defining
\begin{equation*}
\Phi_\phi(X) = \begin{cases}
\ep{\iu\phi\Qh}\Phi(X)\ep{-\iu\phi\Qh}&\text{if }X\cap\eta_{-}\neq\emptyset\text{ and }X\cap\eta_{-}^c\neq\emptyset \\
\Phi(X)&\text{otherwise}
\end{cases}
\end{equation*}
It follows that $H_\phi$ and $H$ differ from each other only along $\eta_{-}$, while $\tilde H_\phi$ and $H_\phi$ differ from each other only along $\eta_+$.

\begin{prop}
Let $U_\parallel$ be the solution of
\begin{equation*}
\partial_s U_\parallel(s) = \iu K_s U_\parallel(s),\qquad U_\parallel(0) = \idtyty,
\end{equation*}
at $s=1$. Here, $K_s$ is given by~(\ref{Ks}) with the concrete Hamiltonian $H_{\phi(s)}$. Let
\begin{equation*}
U = \ep{2\pi\iu (\tilde K_--\Qh)}
\end{equation*}
where
\begin{equation*}
\tilde K_- = \int W(u) \ep{\iu u \tilde H_\phi} (\partial_\phi \tilde H_\phi)_- \ep{-\iu u \tilde H_\phi} du \,\Big\vert_{\phi = 0}
\end{equation*}
Then $U_\parallel = U + \caO(L^{-\infty})$. In particular, (\ref{parallel yields Kubo}) holds.
\end{prop}
\begin{proof}
By construction,
\begin{equation*}
\partial_\phi H_\phi =  (\iu [\Qh, \tilde H_\phi])_- = (\partial_\phi \tilde H_\phi)_-.
\end{equation*}
Hence $\partial_s H_s = (\partial_s \tilde H_s)_-$ so that
\begin{align}
K_s &=\int W(u) \ep{\iu u H_{s}} (\partial_s \tilde H_s)_- \ep{-\iu u H_s} du \nonumber \\
&= \phi'(s)\int W(u) \ep{\iu u \tilde H_{\phi(s)}} (\partial_\phi \tilde H_\phi\vert_{\phi = \phi(s)})_- \ep{-\iu u \tilde H_{\phi(s)}} du + \caO(L^{-\infty}) \nonumber\\
&= \phi'(s) (\tilde K_{\phi(s)})_- + \caO(L^{-\infty}) \label{tilde K- is K-}
\end{align}
for any $s\in[0,1]$. The first and last equalities are just the definition of $K_s$, resp.~$\tilde K_s$ with the observation above. The second follows from the Lieb-Robinson bound using the fact that $\tilde H_s - H_s$ is supported along $\eta_{+}$. Now, $\tilde H_\phi$ is a gauge covariant family, see~(\ref{Twist Antitwist}), hence $\tilde K_\phi = \ep{\iu \phi \Qh} \tilde K \ep{-\iu \phi \Qh}$ and
\begin{equation*}
(\tilde K_\phi)_- = \ep{\iu \phi \Qh} \tilde K_- \ep{-\iu \phi \Qh}.
\end{equation*}
This shows that $\phi'(s) (\tilde K_{\phi(s)})_-$ generates the propagator $\ep{\iu\phi(s) \Qh}\ep{\iu\phi(s)(\tilde K_- - \Qh)}$. We conclude by~(\ref{tilde K- is K-}) and the uniqueness of the solution of ODEs that
\begin{equation*}
U_{\parallel}(s) = \ep{\iu\phi(s) \Qh}\ep{\iu\phi(s)(\tilde K_- - \Qh)} + \caO(L^{-\infty})
\end{equation*}
for all $s\in[0,1]$. In particular,
\begin{equation}\label{U//U}
U_{\parallel}(1) = \ep{2\pi \iu (\tilde K_- - \Qh)} + \caO(L^{-\infty})
\end{equation}
since $\phi(1) = 2\pi$ and by integrality of the spectrum of charge. 

Theorem~3.2 of~\cite{MBIndex} proves the identity~(\ref{parallel yields Kubo}) with $T_-$ associated with $U = \ep{2\pi \iu (\tilde K_- - \Qh)}$ instead of $U_\parallel$, hence (\ref{U//U}) concludes the proof of~(\ref{parallel yields Kubo}).
\end{proof}

\subsection{Lieb-Robinson bounds}\label{Appendix LR}

The Lieb-Robinson bound for a local dynamics $\tau_t$ on $\Gamma$ generated by an extensive observable satisfying the decay condition~(\ref{Extensive}) implies that for $r>0$ and $\mathrm{supp}(O_X) = X$,
\begin{align*}
\Vert \tau_t(O_X) - \bbE_{X_{(r)}}(\tau_t(O_X))\Vert
&\leq \int_{\caU(X_{(r)}^c)}\Vert [\tau_t(O_X),U]\Vert d\mu(U)\\
&\leq \Vert O_X\Vert\vert X \vert \ep{\xi t}f(r)
\end{align*}
see~\cite{LRYoshiko}, where $\bbE_{Z}$ denotes the normalized partial trace over $Z^c = \Gamma\setminus Z$, and $\caU(Z^c)$ is the unitary group in the algebra of observables supported in $Z^c$, equipped with its Haar measure $\mu$. The function $f$ and the constant $\xi>0$ depend on the generator of~$\tau$ but neither on $\Gamma$ nor on the observable $O_X$. As already noted, $f$ is a positive, decreasing function decaying faster than any inverse power. {While the Lieb-Robinson bound in this form is useful in any fixed time interval $t\in[0,T]$, it is desirable to have time-independent error bounds in the present adiabatic context. }

Let us first consider the case $f(r) = C \ep{-\zeta r}$ for some $\zeta>0$, which happens for a dynamics generated by an extensive observable satisfying the finite range condition (i) of Section~\ref{subsec: extensive}. By picking $r = v t + \delta$ with $v = \frac{\xi}{\zeta}$, we obtain
\begin{equation*}
\Vert \tau_t(O_X) - \bbE_{X_{(vt + \delta)}}(\tau_t(O_X))\Vert
\leq C\Vert O_X \Vert\vert X \vert \ep{-\zeta \delta},
\end{equation*}
where the constant $C$ is uniform in $t$, namely~(\ref{LR bound}).

In the present context, the generator $K$ as well as the $A_j$'s and hence also $R_{n,\epsilon}$ all have slower decay, expressed concretely as $D_n(r) = r^k\ep{-C r/\ln(\ln(r))}$ for some $k = k(n)$, see~\cite{MathAdiabatic}. No affine choice $r(t)$ as above will be such that $\ep{\xi t}D(r(t))$ is uniformly bounded on $[0,\infty)$. In order to deal with that in a rather explicit fashion, we note that $D_n(r) \leq C\ep{-\zeta r^\beta}$ for any $0<\beta<1$ and any $\zeta>0$. In this case, we pick
\begin{equation*}
r = c t^{1/\beta} + \delta,
\qquad c = \frac{1}{2}\left(\frac{2\xi}{\zeta}\right)^{1/\beta}.
\end{equation*}
By midpoint concavity,
\begin{equation*}
\zeta r^\beta \geq \xi t + 2^{\beta-1}\zeta \delta^\beta
\end{equation*}
and hence
\begin{equation*}
\ep{\xi t} D_n(r) \leq C\ep{\xi t}\ep{-\zeta r^\beta}\leq C\ep{-2^{\beta-1}\zeta \delta^\beta} = \caO(\delta^{-\infty}).
\end{equation*}
We conclude that, for a dynamics $\sigma_t$ generated by a generalized extensive observable that satisfies~(\ref{Extensive}) with only a subexponential $f$, we have 
\begin{equation*}
\Vert \sigma_t(O_X) - \bbE_{X_{(c t^{1/\beta} + \delta)}}(\sigma_t(O_X))\Vert
\leq C \Vert O_X\Vert\vert X \vert \ep{-2^{\beta-1}\zeta \delta^\beta}
\end{equation*}
where the constants $c,C,\zeta$ are again uniform in $t$. This is~(\ref{Slow LR}) with $\alpha = 1/\beta$ and $\mu = 2^{\beta-1}\zeta$.



\end{document}